\documentclass{article}
\usepackage[a4paper,margin=1in]{geometry}
\usepackage{setspace}
\usepackage{graphicx} 
\usepackage{amsmath}
\usepackage{mathtools}
\usepackage{amsthm}
\usepackage{tablefootnote}
\usepackage{braket}
\usepackage{amssymb}
\usepackage{mathrsfs}
\usepackage{amsfonts}
\usepackage{authblk}
\usepackage{makecell}
\usepackage{threeparttable}
\usepackage{bbding}
\usepackage{subfig}
\usepackage{dcolumn}
\usepackage{float}
\usepackage{comment}
\usepackage{multirow}
\usepackage{color}
\usepackage{tcolorbox}
\usepackage{algorithm}
\usepackage{algorithmic}
\usepackage{tikz}
\usepackage{appendix}
\usepackage[linktocpage=true, colorlinks,
linkcolor=blue,citecolor=blue,
bookmarks,bookmarksopen,bookmarksnumbered]
{hyperref}
\usepackage{circuitikz} 
\usepackage{tikz-3dplot} 
\usepackage{adjustbox}
\allowdisplaybreaks[4]

\newcommand{\rbra}[1]{\left( #1 \right)} 
\newcommand{\sbra}[1]{\left[ #1 \right]}
\newcommand{\cbra}[1]{\left\{ #1 \right\}}
\newcommand{\abs}[1]{\lvert #1 \rvert}
\newcommand{\Abs}[1]{\lVert #1 \rVert}

\theoremstyle{plain}

\newtheorem{theorem}{Theorem}[section]

\newtheorem{lemma}[theorem]{Lemma}

\theoremstyle{definition}
\newtheorem{definition}[theorem]{Definition}

\theoremstyle{remark}

\begin{document}
\title{Quantum Speedups for Testing Similar Means}
\author[1]{Chengshen Gao}
\author[2]{Yongzhen Xu}
\author[2]{Shenggen Zheng}
\author[1,2]{Lvzhou Li\thanks{Corresponding author: lilvzh@mail.sysu.edu.cn}}
\affil[1]{Institute of Quantum Computing and Computer Theory, School of Computer Science and Engineering, Sun Yat-sen University, Guangzhou 510006, China}
\affil[2]{Quantum Science Center of Guangdong-Hong Kong-Macao Greater Bay Area, Shenzhen, 518045, Guangdong, China}
\date{}

\maketitle

\begin{abstract}
    Property testing of distributions is a central topic in information theory, learning theory, and statistics. While quantum algorithms are known to offer significant speedups for property testing of a single distribution or a pair of distributions, it is unclear whether quantum algorithms provide speedups for property testing of $m$ ($m\geq 3$) distributions. This work focuses on quantum algorithms for testing whether $m$ distributions have similar means or are $\epsilon$-far from mean similarity under two models. In the query model, the algorithm can choose which distribution to sample from, whereas in the sampling model, the distributions are selected uniformly. We design quantum algorithms with complexities $\tilde{O}(1/\epsilon)$ (the $\tilde{O}$  notation hides poly-logarithmic factors) and $\tilde{O}(\sqrt{m}/\epsilon)$ in the query and sampling models, respectively, achieving quadratic speedups over the classical counterparts. We further establish quantum lower bounds of $\Omega\left(1/\epsilon\right)$ and $\Omega\rbra{m^{1/3}+\frac{m^{1/4}}{\epsilon}}$ for the query model and the sampling model, demonstrating the optimality of our quantum algorithms in terms of the dependence on $\epsilon$ up to logarithmic factors.
\end{abstract}

\section{Introduction}

One of the main motivations for work in quantum information science is the prospect of fast quantum algorithms to solve important computational problems \cite{montanaro2016quantum,zhang2022brief}. Since the first quantum algorithm, Deutsch's algorithm \cite{Deutsch1985}, was proposed in 1985,  significant efforts have been made in quantum algorithm design over the past forty years. Beyond the well-known quantum speedups in foundational problems such as factorization \cite{Shor1994}, unordered database search \cite{Grover1996}, and solving linear equations \cite{HHL2009}, quantum algorithms have also demonstrated advantages in solving combinatorial optimization problems \cite{jordan2025optimization}, NP search problems \cite{YamakawaZ22}, the welded tree problem \cite{ChildsCDFGS03,jeffery2023multidimensional,li2024recovering-a,li2024recovering-b}, spatial search problems \cite{ quadratic_20, quadratic_22, li2025quantum,li2026deterministic,wang2025unifying}, geometric problems \cite{li2025unbounded, nghiem2025quantum, andrejevs2026quantum}, string learning problems \cite{xu2025provable, qi2026quantum, xu2023quantum, AmbainisM14, CleveIGNTTY12}, problems with abstract structures (e.g., groups \cite{childs2010quantum, kuperberg2005subexponential,  watrous2001quantum,  magniez2007quantum, ye2022deterministic}, matroids \cite{huang2024quantum-a,huang2024quantum-b,huang2026quantum}), and others. 

In the context of distribution property testing, quantum algorithms are known to offer significant speedups for basic tasks such as uniformity testing and closeness testing \cite{chakraborty2010new,bravyi2011quantum,montanaro2015quantum,gilyen2020distributional,luo2024succinct,canonne_et_al:LIPIcs.TQC.2025.7}. Property testing of distributions~\cite{rubinfeld2012taming,canonne2020survey} is a central topic in information theory, learning theory, and statistics. More specifically, the algorithm is given sample access to one or more unknown distributions and aims to determine whether they satisfy a global property $\mathcal{P}$ or are far from satisfying it.

Property testing of quantum states can be viewed as a generalization of quantum distribution property testing, since classical distributions correspond to quantum states that are diagonal in the computational basis. Research on property testing of quantum states addresses a number of fundamental problems, notably quantum state discrimination\cite{chefles2000quantum,barnett2009quantum,bae2015quantum},  productness testing\cite{harrow2013testing,soleimanifar2022testing}, mixedness testing\cite{childs2007weak,o2015quantum} and closeness testing\cite{buhrman2001quantum,gilyen2022improved,liu2025quantum}.

In contrast to the aforementioned studies which focus on testing properties of one or two objects in the corresponding model, property testing for a \emph{collection} of $m$ distributions or quantum states constitutes a fundamentally important and significantly more challenging problem.  These considerations have already spurred investigation into several specific problems of classical distribution property testing, such as equivalence testing~\cite{levi2013testing, diakonikolas2016new, diakonikolas2021optimal}, clusterability testing~\cite{levi2013testing, levi2014testing}, and similar means testing~\cite{levi2014testing}. In the setting of quantum state property testing, equivalence testing for multiple quantum states has also been studied~\cite{yu2021sample,fanizza2023testing}.

However, to the best of our knowledge, no quantum algorithms are known for the task of testing similar means of a collection of distributions. Formally, testing similar means is defined as follows. Let $\mathcal{D} = (D_1,\ldots,D_m)$ be a collection of distributions over $\{0,\ldots,n\}$, and let $\gamma,\varepsilon > 0$ be parameters. 
We say that $\mathcal{D}$ has $\gamma$-similar means
if there exists an interval $I \subseteq \mathbb{R}$ of length at most $\gamma n$ such that $\mathbb{E}[D_i] \in I $ for all $ i \in [m]$. Conversely,  $\mathcal{D}$ is said to be $\epsilon$-far from having this property if  for every collection $\mathcal{D}^{*} = (D_1^{*},\ldots,D_m^{*})$ that has $\gamma$-similar means, it holds that $\frac{1}{m} \sum_{i=1}^m d(D_i, D_i^{*}) > \epsilon$, where $d(\cdot,\cdot)$ is some distance measure between distributions. The goal of testing $\gamma$-similar means is then to distinguish collections that have $\gamma$-similar means from those that are $\varepsilon$-far from this property. 

This problem has been fully characterized in the classical setting under two models \cite{levi2014testing}. In the query model, the algorithm may indicate an index $i\in\sbra{m}$ of its choice and receive a sample $j$ drawn according to the distribution $D_i$. In the sampling model, the algorithm obtains a pair $(i,j)$, where $i$ is selected uniformly and $j$ is drawn according to $D_i$. In terms of the query model, classical methods achieve a query complexity of $\tilde{O}\left(\frac{1}{\epsilon^2}\right)$ with a matching lower bound $\Omega\left(\frac{1}{\epsilon^2}\right)$, while their complexity in the sampling model has an upper bound of $m^{1-\Omega\rbra{\epsilon^2}}\cdot\tilde{O}\rbra{\frac{1}{\epsilon^2}}$ and a lower bound of $(1-\gamma)m^{1-\tilde{O}\left(\sqrt{\frac{\epsilon}{\gamma}}\right)}$ when $\epsilon=\Theta(1)$. This naturally leads us to the following fundamental question.
\begin{center}
    \emph{Can quantum algorithms speed up testing mean similarity?}
\end{center}

In this work, we answer this question in the affirmative by giving quantum algorithms for testing mean similarity of multiple distributions that provably outperform their classical counterparts in both the query and sampling models. To state our results, we consider two natural quantum access models: quantum query access and quantum sampling access, which generalize the classical query and sampling models, respectively. Formal definitions are given below.

\begin{table*}[!b]
\centering
\caption{Classical and Quantum Complexities in Query and Sampling Models}
\label{tab: summary}
{\footnotesize
\begin{tabular*}{\linewidth}{cccc}
	\hline
	\multicolumn{2}{c}{}&Classical\cite{levi2014testing} & Quantum \\
	\hline
	& & &  \\[-5pt]
	\multirow{4}*{Query model}&Upper bound&$\tilde{O}\rbra{\frac{1}{\epsilon^2}}$&$\tilde{O}\rbra{\frac{1}{\epsilon}}$(Theorem~\ref{thm: upper bound of Oracle A})\\[5pt]
	
	& & &  \\[-5pt]
	&Lower bound&$\Omega\rbra{\frac{1}{\epsilon^2}}$&$\Omega\rbra{\frac{1}{\epsilon}}$(Theorem~\ref{thm: lower bound of Oracle A})\\[5pt]
	
	& & &  \\[-5pt]
	\multirow{4}*{Sampling model}&Upper bound&$m^{1-\Omega\rbra{\epsilon^2}}\cdot\tilde{O}\rbra{\frac{1}{\epsilon^2}}$&$\tilde{O}\rbra{\frac{\sqrt{m}}{\epsilon}}$(Theorem~\ref{thm: upper bound of Oracle B})\\[5pt]
	
	& & &  \\[-5pt]
	&Lower bound&$\rbra{1-\gamma}m^{1-\tilde{O}\rbra{\sqrt{\frac{\epsilon}{\gamma}}}}$&$\Omega\rbra{m^{1/3}+\frac{m^{1/4}}{\epsilon}}$(Theorem~\ref{thm: lower bound of Oracle B})\\[5pt]
	\hline
\end{tabular*}
\\\vspace{1mm}\parbox{\textwidth}{Note: $m$ denotes the number of distributions $D_i$ in the collection. $n$ represents the size of the support of each distribution. $\gamma$ is a similarity parameter that controls the proximity of the means. $\epsilon$ quantifies the deviation from this similarity condition.}
}
\end{table*}

\begin{definition}[Quantum query access]
    A collection of distributions $\mathcal{D} = (D_1,\ldots,D_m)$ over a common domain $\sbra{0,...,n}$ is said to have quantum query access if we are given access to a unitary $O_A$ (together with its inverse and controlled versions)  that satisfies
    \begin{align}
        O_A\ket{i}_1\ket{\vec{0}}_{2,3}=\ket{i}_1\left(\sum_{j=0}^{n} \sqrt{D_i\rbra{j}} \ket{j}_2\ket{\varphi_{i,j}}_3\right),
    \end{align} 
    for every $i\in\sbra{m}$. Here each $D_i:\cbra{0,...,n}\to[0,1]$ is a probability distribution (i.e., $\sum_{j=0}^nD_i\rbra{j}=1$).  
    $\cbra{\ket{\varphi_{i,j}}}$ are arbitrary normalized quantum states obeying the orthogonality condition $\langle \varphi_{i,j}|\varphi_{i,l}\rangle=\delta_{jl}$.
\end{definition}

\begin{definition}[Quantum sampling access]
   A collection of distributions $\mathcal{D} = (D_1,\ldots,D_m)$ over a common domain $\sbra{0,...,n}$ is said to have quantum sampling access if we are given access to a unitary $O_B$ (including its inverse and controlled versions) that fulfills
    \begin{align}
    \begin{split}
        &O_B\ket{\Vec{0}}_{1}\ket{\Vec{0}}_{2,3}=\\
        &\sum_{i=1}^{m}\frac{1}{\sqrt{m}}\ket{i}_1\left(\sum_{j=0}^n \sqrt{D_i\rbra{j}}\ket{j}_2\ket{\varphi_{i,j}}_3\right), 
    \end{split}
    \end{align}
where $\cbra{\ket{\varphi_{i,j}}}$ denote arbitrary normalized quantum states satisfying $\langle \varphi_{i,j}|\varphi_{i,l}\rangle=\delta_{jl}$.
\end{definition}

Under the two quantum access models defined above, we summarize our main results below, see also Table~\ref{tab: summary}.

\noindent  \textbf{Results in the quantum query access model.}
In the \emph{query model}, we prove that our algorithm attains a quadratic speedup in terms of $\epsilon$, achieving a query complexity of $\tilde{O}\rbra{\frac{1}{\epsilon}}$. Furthermore, we establish a lower bound of $\Omega\rbra{\frac{1}{\epsilon}}$, showing that the algorithm is asymptotically optimal up to logarithmic factors.

\noindent  \textbf{Results in the quantum sampling access model.}
For the \emph{sampling model}, we give a quantum algorithm that attains a quadratic speedup in the dependence on both the accuracy parameter $\epsilon$ and the collection size $m$, with a query complexity of $\tilde{O}\rbra{\frac{\sqrt{m}}{\epsilon}}$. For the lower bound, we also prove a lower bound of $\Omega\rbra{m^{1/3}+\frac{m^{1/4}}{\epsilon}}$ for this setting.

\section{Preliminaries} \label{Preliminaries}

\subsection{Notations}

$\mathcal{D} = (D_1,\ldots,D_m)$ is a collection of distributions over $\{0,\ldots,n\}$, where $D_i:\cbra{0,...,n}\rightarrow[0,1]$ and $\sum_{j=0}^nD_i\rbra{j}=1$. The distance measure $d\rbra{\cdot,\cdot}$ is the $\ell_1$ distance in this work, where the $\ell_\alpha$ distance is defined as 
\begin{align*}
    d\rbra{D_1,D_2}:=\rbra{\sum_{j=0}^n\abs{D_1\rbra{j}-D_2\rbra{j}}^\alpha}^\frac{1}{\alpha}.
\end{align*}
For a distribution $D$ over $\{0,\ldots,n\}$, the mean of $D$ is defined as $\mu(D) := \sum_{i=0}^{n} i \cdot D(i)$ and we denote $\mu_i$ as $\mu\rbra{D_i}$. For a value $0 \le z \le n$, let $d_1(D,z) := \min_{D' : \mu(D') = z} \, \lVert D - D' \rVert_1$ denote the minimum $\ell_1$ distance between $D$ and a distribution that has mean $z$.

\subsection{Testing Similar Means}
Here, we formally define the fundamental concepts introduced in \cite{levi2014testing} and present several useful lemmas for testing similar means.

\begin{definition}[Extremal feasible means]
    \label{def: extremal mean}
    For a distribution $D$ and a distance parameter $0 \le \beta \le 1$, the quantities $\mu_\beta^{>}(D)$ and $\mu_\beta^{<}(D)$ capture the largest and smallest means attainable by any distribution within $\ell_1$-distance $\beta$ of $D$.
    \begin{itemize}
        \item Define $\mu_\beta^{>}(D)$ as follows: If $d_1(D, n) \ge \beta$, then let $\mu_\beta^{>}(D)$ be the value $\mu>\mu(D)$ such that $d_1(D, \mu) = \beta$. Otherwise, set $\mu_\beta^{>}(D) = n$.
        
        \item Define $\mu_\beta^{<}(D)$ as follows: If $d_1(D, 0) \ge \beta$, then let $\mu_\beta^{<}(D)$ be the value $\mu < \mu(D)$ such that $d_1(D, \mu) = \beta$. Otherwise, set $\mu_\beta^{<}(D) = 0$.
    \end{itemize}
\end{definition}

\begin{lemma}\label{lem: distance sum}
    $\mu_\beta^{>}(D)=\sup\{\mu(D'):\|D'-D\|_1\leq \beta\}$, $\mu_\beta^{<}(D)=\inf\{\mu(D'):\|D'-D\|_1\leq \beta\}$. 
\end{lemma}

\begin{proof}
     We prove the statement for $\mu_\beta^{>}(D)$. The argument for $\mu_\beta^{<}(D)$ is analogous. Define the function $g(z) := d_1(D, z)$. Then we observe that:
\begin{enumerate}
    \item $g(\mu(D)) = 0$, and $g(z) \ge 0$ for all $z$;
    \item $g$ is a continuous and convex function, and hence is non-decreasing for all $z \ge \mu(D)$.
\end{enumerate}
Therefore, there exists some $z \in [\mu(D), n]$ such that $g(z) = \beta$.
Let $z^* := \max \{ z : g(z) = \beta \}$. For any $z > z^*$ and any distribution $D'$ with $\mu(D') = z$, we have $\lVert D' - D \rVert_1 \ge g(z) > g(z^*) = \beta.$
It follows that $z^*=\sup \{ \mu(D') : \lVert D' - D \rVert_1 \le \beta \}$
\end{proof}
Lemma~\ref{lem: distance sum} provides an equivalent characterization of the extremal range of means achievable under $\ell_1$-bounded perturbations of $D$, as captured by Definition~\ref{def: extremal mean}. This viewpoint turns distributional distance into a convenient tool for analyzing mean deviations, and will be used repeatedly in the following sections. Building on this, the next lemma characterizes collections that are $\varepsilon$-far from having $\gamma$-similar means.

\begin{lemma}[Lemma 2 of \cite{levi2014testing}]\label{lem: property of far from}
Let $\mathcal{D}$ be a collection of distributions.
If $\mathcal{D}$ is $\epsilon$-far from having $\gamma$-similar means,
then there exists an interval $[x,y]\subseteq [n]$ with $y-x=\gamma n$
such that
\begin{align}
\begin{split}
    \sum_{i:\,\mu(D_i)>y} d_1(D_i,y) > (\epsilon/4)m,\\
    \sum_{i:\,\mu(D_i)<x} d_1(D_i,x) > (\epsilon/4)m .
\end{split}
\end{align}
\end{lemma}

\subsection{Quantum Subroutines}
We first introduce quantum amplitude estimation, which estimates a probability with quadratic speedup over classical sampling.
\begin{lemma}[Quantum amplitude estimation, Theorem 12 in \cite{brassard2002quantum}]\label{lem: amplitude estimation}
    $A$ is a unitary operator and $P$ is a projection operator, such that $A|0\rangle = \sqrt{a}|\Psi_0\rangle + \sqrt{1-a}|\Psi_1\rangle$, where $a \in [0,1]$ and $P|\Psi_0 \rangle=|\Psi_0 \rangle$, $P|\Psi_1 \rangle=0$. For an integer $t > 0$ and $\delta\in\rbra{0,1}$, the quantum amplitude estimation algorithm $\text{QAE}\rbra{A\ket{0},P,t,\delta}$ outputs $\ket{\hat{a}}$. Upon measuring this output, the measurement result $\hat{a}$ satisfies
    \begin{align*}
        \Pr\left[|\hat{a} - a| \leq \frac{\sqrt{a(1-a)}}{t} + \frac{1}{t^2}\right] \geq 1-\delta.
    \end{align*}
    The algorithm uses $O(t\log\frac{1}{\delta})$ applications of $A$ and $P$.
\end{lemma}

Next, we present quantum mean estimation which can estimate the mean of a random variable that has the quantum codes, and is quadratically faster than classical sampling.

\begin{lemma}[Quantum mean estimation, Theorem 1.1 in \cite{kothari2023mean}]\label{lem: mean estimation}
    Let $O_\Omega$ and $O_X$ be the quantum codes of a random variable $X$ on a finite probability space such that 
\begin{align*}
\begin{split}
    &O_\Omega\ket{\Vec{0}}=\sum_{\omega\in\Omega}\sqrt{p(\omega)}\ket{\omega}\ket{\psi_\omega},\\
    &O_X\ket{\omega}\ket{\Vec{0}}=\ket{\omega}\ket{x_\omega}\ket{\Vec{0}},
\end{split}
\end{align*}
where $\Omega$ is the domain of the probability distribution of $X$, $p:\Omega \to\sbra{0,1}$ 
is the probability distribution of $X$ satisfying $p(\omega) \ge 0$ for all 
$\omega \in \Omega$ and $\sum_{\omega \in \Omega} p(\omega) = 1$, and $\cbra{\ket{\psi_\omega}}$ are arbitrary normalized quantum states. Given $\epsilon,\delta\in\rbra{0,1}$, the quantum mean estimation algorithm QME$\rbra{O_\Omega, O_X,\epsilon,\delta}$ outputs an estimation $\tilde{\mu}$ such that $\Pr[|\tilde{\mu}-\mu|\geq \sigma\epsilon]\leq \delta$, where $\mu=\mathbb{E}\sbra{X}$ and $\sigma$ are the mean and the standard deviation of the random variable $X$, using $O\rbra{\frac{\log\frac{1}{\delta}}{\epsilon}}$ applications of $O_X$ and $O_X^\dagger$.
\end{lemma}

We next describe the exact amplitude amplification algorithm that deterministically prepares the target state if the amplitude of the target state is known.

\begin{lemma}[Quantum deterministic amplitude amplification\cite{brassard2002quantum,hoyer2000arbitrary,long2001grover,li2023deterministic}]\label{lem: deterministic amplitude amplification}
Suppose there exists a unitary $U$ such that
\[
U\ket{\vec{0}} = \sqrt{a}\ket{\psi_0} + \sqrt{1-a}\ket{\psi_1},
\]
where \(a>0\) is known and \(\ket{\psi_0}, \ket{\psi_1}\) are orthogonal quantum states.  
Then there exists a quantum algorithm that deterministically prepares the state \(\ket{\psi_0}\) using \(\Theta(1/\sqrt{a})\) applications of \(U\) and \(U^{\dagger}\).
\end{lemma}

\section{Results for Quantum Query Access Model}

In this section, we study testing similar means with the quantum query access. We propose a quantum algorithm, establish its correctness, and show that it achieves a query complexity of $\tilde{O}\rbra{\frac{1}{\epsilon}}$. Furthermore, we show that this complexity is optimal up to logarithmic factors by proving a corresponding lower bound.
\subsection{Upper Bound}
\textbf{Limitation of the direct approach.}
A natural strategy for testing whether $\mathcal{D}$ is $\epsilon$-far from having similar means is to sample a subset of the distributions and estimate their means directly. By Lemma~\ref{lem: property of far from}, if $\mathcal{D}$ is $\epsilon$-far from having $\gamma$-similar means, then there exists an interval $[x,y]$ of length $\gamma n$ that separates the collection: $\Omega(\epsilon m)$ distributions have means below $x-\Omega(\epsilon n)$, and another $\Omega(\epsilon m)$ have means above $y+\Omega(\epsilon n)$. In particular, it suffices to inspect only $O(1/\epsilon)$ distributions. Applying Lemma~\ref{lem: mean estimation} to each of them yields a total query complexity of $O(1/\epsilon^2)$.

\textbf{Why the direct approach fails.} 
The above analysis shows that quantum mean estimation alone is not sufficient for testing similar means. 
The underlying reason is that similarity in means does not imply similarity under $\ell_1$ distance: in general, $|\mu(D)-x|$ and $d_1(D,x)$ are not equivalent. 
A distribution may have its mean arbitrarily close to $x$, yet be $\ell_1$-far from a distribution with mean exactly $x$.
Consequently, the approach that uses the mean difference $|\mu(D)-x|$ as a surrogate for $d_1(D,x)$ is not sufficient for the bucketing argument.

\textbf{The quantum feasible mean estimation under $\ell_1$-bounded perturbations ($\ell_1$QFME).} 
To address this, our algorithm avoids estimating $\mu(D)$ directly. Instead, following~\cite{levi2014testing}, we design Algorithm~\ref{alg: QuantumGetBounds} to estimate the extremal feasible means $\mu_{\beta}^{>}(D)$ and $\mu_{\beta}^{<}(D)$ from Definition~\ref{def: extremal mean}. More precisely, given the quantum query access to a distribution $D_i$ and a distance parameter $\beta$, the subroutine estimates the means of two distributions $D_i^{>}$ and $D_i^{<}$ satisfying $d(D_i^{>},D_i)=d(D_i^{<},D_i)=\beta$, while also satisfying
\begin{align*}
\begin{split}
    \mu(D_i)\le \tilde{\mu}(D_i^{>})\le \mu^{>}_{\beta}(D_i),\\
    \mu^{<}_{\beta}(D_i)\le \tilde{\mu}(D_i^{<})\le \mu(D_i),
 \end{split}
\end{align*}
where $|\tilde{\mu}(D_i^{>})-\mu(D_i^{>})|\leq O\rbra{\beta n}$, $|\tilde{\mu}(D_i^{<})-\mu(D_i^{<})|\leq O\rbra{\beta n}$. 

This quantum procedure controls the deviation of a distribution's mean under $\ell_1$ perturbations, providing the structural bounds needed for our quantum testing procedure. We proceed to describe Algorithm~\ref{alg: QuantumGetBounds}, followed by proofs of its correctness.
\begin{algorithm}[H]
\caption{$\ell_1$QFME$\rbra{i, \beta,\delta}$}
\label{alg: QuantumGetBounds}
\begin{algorithmic}[1]
	\STATE {\bfseries Input:} Quantum query access $O_A$ to a collection $\mathcal{D}=(D_1,\ldots,D_m)$, the index $i$, parameters $\beta,\delta\in\rbra{0,1}$.
	\STATE {\bfseries Output:} $e$ and $f$.
    \STATE $d\gets 100\lceil\log(12/\delta)/\beta\rceil$, $t\gets\left\lceil \frac{6}{\beta} \right\rceil$.
	\STATE Prepare $d$ copies of $\ket{\psi_0}=O_A \ket{i}_1 \ket{\vec{0}}_{2,3}$.
    \STATE Measure register 2 of each copy, sort outcomes: $j_1\le\cdots\le j_d$.
    \STATE $a \gets j_{\frac{\beta d}{4}}$, $b \gets j_{(1-\frac{\beta}{4}) d}$, $a'\gets \frac{a+n}{2}$, $b' \gets \frac{b}{2}$.
    \STATE $P_a=I \otimes \ket{a}\bra{a}_2$, $P_b=I \otimes \ket{b}\bra{b}_2$.
    \STATE $\tilde{p}(a)\gets \text{QAE}\rbra{\ket{\psi_0},P_a,t,\frac{\delta}{6}}$, \\$\tilde{p}(b)\gets \text{QAE}\rbra{\ket{\psi_0},P_b,t,\frac{\delta}{6}}$ by Lemma~\ref{lem: amplitude estimation}.
    \STATE If $\tilde{p}(a) \le \beta/4$, $a''\gets a'$. Otherwise $a''\gets \frac{\beta}{4\tilde{p}(a)} a' + \left(1 - \frac{\beta}{4\tilde{p}(a)}\right) a$. 
    \STATE If $\tilde{p}(b) \le \beta/4$, $b''\gets b'$. Otherwise, $b''\gets \frac{\beta}{4\tilde{p}(b)} b' + \left(1 - \frac{\beta}{4\tilde{p}(b)}\right) b$.
    \STATE Let X be a variable such that $X \sim D_i$. Construct $O_\Omega$, $O_Y$, $O_Z$ which are the quantum codes of the following two variables $Y$ and $Z$ based on $X$ with $O_A$:
    \begin{align*}
        Y:=
            \begin{cases}
                a',
                &  X<a, \\[0.3em]
               a''
                &  X = a, \\[0.3em]
                X
                &  X>a. \\[0.3em]
                \end{cases}
        Z:=
            \begin{cases}
                X &  X<b, \\[0.3em]
               b''
                & X = b, \\[0.3em]
                b'
                &  X>b.\\[0.3em]
                \end{cases}
    \end{align*}
    \STATE $e\gets \text{QME}\rbra{O_\Omega, O_Y, \beta, \delta/6}$, \\$f\gets \text{QME}\rbra{O_\Omega, O_Z, \beta, \delta/6}$ by Lemma~\ref{lem: mean estimation}.
    \STATE {\bfseries Return} $e$ and $f$.
\end{algorithmic}
\end{algorithm}

\begin{theorem}\label{thm:l1qeme}
    For $i\in[m]$ and $\beta,\delta\in(0,1)$ and given the quantum query access, Algorithm~\ref{alg: QuantumGetBounds} solves the $\ell_1$QFME problem for the distribution $D_i$ with failure probability at most $\delta$.
Specifically, it outputs values $e$ and $f$ such that
\begin{align}
\begin{split}
    \mu(D_i)\le e\le \mu^{>}_{\beta}(D_i),\\
    \mu^{<}_{\beta}(D_i)\le f\le \mu(D_i),
\end{split}
\end{align}
using $O\rbra{\frac{\log\frac{1}{\delta}}{\beta}}$ queries to $O_A$ and $O_A^\dagger$.
\end{theorem}

\begin{proof}
    We prove that $e$ satisfies $\mu(D_i) \le e \le \mu_\beta^{>}(D_i)$, and the proof for $f$ follows analogously.
    
\textbf{Correctness.} For any $p\in\sbra{0,1}$, the quantile of $D_i$ is an index $Q_{D_i}\rbra{p}\in\sbra{n}$ satisfying $\Pr_{D_i}[j < Q_{D_i}\rbra{p}] < p$ and $\Pr_{D_i}[j \leq Q_{D_i}\rbra{p}] \geq p$. Let $X_p$ denote the number of outcomes in which register~2 is measured at Step~4 such that $j \le Q_{D_i}\rbra{p}$, and let $Y_p$ denote the number of measurements of register 2 at Step 4 such that $j<Q_{D_i}\rbra{p}$. Therefore, 
\begin{align*}
\begin{split}
    \mathbb{E}\sbra{X_p}=d\Pr_D[j \leq Q_{D_i}\rbra{p}]\geq dp,\\
    \mathbb{E}\sbra{Y_p}=d\Pr_D[j < Q_{D_i}\rbra{p}]<dp.
\end{split}
\end{align*}
The multiplicative Chernoff bound shows that
\begin{align*}
\begin{split}
    &\Pr[X_p\leq \rbra{1-\alpha}\mathbb{E}\sbra{X_p}]\leq \rbra{\frac{e^{-\alpha}}{\rbra{1-\alpha}^{1-\alpha}}}^{\mathbb{E}\sbra{X_p}},\\
     &\Pr[Y_p\geq \rbra{1+\alpha}\mathbb{E}\sbra{Y_p}]\leq e^{-\frac{\alpha^2\mathbb{E}\sbra{Y_p}}{2+\alpha}}.
\end{split}
\end{align*}

Take $p=\beta/3$ and $\alpha=\frac{1}{4}$. Applying the multiplicative Chernoff bound, we have $\Pr\left[ X_{\beta/3}<\frac{d\beta}{4}\right]<e^{-\frac{d\beta}{100}}$. Since $d = 100\lceil\log(12/\delta)/\beta\rceil=\Theta\rbra{\frac{\log\frac{1}{\delta}}{\beta}}$, the above probability is at most $\delta/12$.

Next, take $p = \beta/8$ and any $\alpha\geq 1$. Applying the multiplicative Chernoff bound again, we obtain $\Pr\left[Y_{\beta/8}\geq \frac{d\beta}{4}\right]\leq e^{-\frac{\alpha\mathbb{E}\sbra{Y_{\beta/8}}}{3}}\leq e^{-\frac{d\beta}{24}}$, where the last inequality follows from the fact that $\alpha\mathbb{E}\sbra{Y_{\beta/8}}=\frac{d\beta}{4}-\mathbb{E}\sbra{Y_{\beta/8}}\geq \frac{d\beta}{8}$. Thus the above probability is at most $\delta/12$. By a union bound, with probability at least $1 - \delta/6$, neither of the above events occurs. 

Conditioned on this event, the outcomes contain at least $\frac{d\beta}{4}$ points with
$j \le Q_{D_i}\rbra{\frac{\beta}{3}}$ and fewer than $\frac{d\beta}{4}$ points with $j<Q_{D_i}\rbra{\frac{\beta}{8}}$. By the choice of $a$ in Step~5 of the algorithm, it follows that $Q_{D_i}\rbra{\frac{\beta}{8}} \le a \le Q_{D_i}\rbra{\frac{\beta}{3}}$.
Therefore,
\begin{align*}
\begin{split}
    &\Pr_{D_i}[j < a] \le \Pr_{D_i}\sbra{j < Q_{D_i}\rbra{\frac{\beta}{3}}} \le \frac{\beta}{3},\\
    &\Pr_{D_i}[j \le a] \ge \Pr_{D_i}\sbra{j \le Q_{D_i}\rbra{\frac{\beta}{8}}} \ge \frac{\beta}{8}.
\end{split}
\end{align*}

With Lemma~\ref{lem: amplitude estimation}, we have $\abs{p\rbra{a}-\tilde{p}\rbra{a}}\leq \frac{\beta}{6}+\frac{\beta^2}{36}\leq \frac{\beta}{4}$
with probability at least $1 - \delta/6$, where $p\rbra{a}=\Abs{P_aO_A\ket{i}_1\ket{\vec{0}}_{2,3}}^2$.

If $\tilde{p}\rbra{a}\leq \frac{\beta}{4}$, then $\Pr[j\leq a]\leq \beta$. Then 
\begin{align*}
\begin{split}
      \mu\rbra{Y}&=a'\Pr_{D_i}[j\leq a]+\sum_{j>a}j\times p\rbra{j}\\
      &\geq \mu\rbra{D_i}+\frac{\rbra{n-a}\beta}{16}.
\end{split}
\end{align*}
With Lemma~\ref{lem: distance sum}, we obtain
\begin{align*}
\begin{split}
    \mu\rbra{Y}&\leq \mu_\beta^>\rbra{D_i}-\frac{n-a}{2}\Pr_{D_i}[j\leq a]\\
    &\leq \mu_\beta^>\rbra{D_i}-\frac{\rbra{n-a}\beta}{16}.
\end{split}
\end{align*}

If $\tilde{p}\rbra{a}>\frac{\beta}{4}$, then $h(a):=\Pr_{D_i}[j<a]+\frac{\beta}{4\tilde{p}\rbra{a}}\leq \beta$. Then
\begin{align*}
\begin{split}
    \mu\rbra{Y}=&a'\Pr_{D_i}[j<a]+a''p\rbra{a}+\sum_{j>a}j\times p\rbra{j}\\
    \geq &\mu\rbra{D_i}+\frac{n-a}{2}h(a), \\
    \mu\rbra{Y}<&\mu_\beta^>\rbra{D_i}-\frac{n-a}{2}h(a).
\end{split}
\end{align*}
One can verify that $h(a)\geq \frac{\beta}{20}$. By Lemma~\ref{lem: mean estimation}, we have $\abs{e-\mu\rbra{Y}}\leq \frac{n\beta}{20}$, with probability at least $1 - \delta/6$. Therefore, $\mu\rbra{D_i}\leq e\leq \mu_\beta^{>}(D_i)$.

The total failure probability of $e$ is at most $\frac{\delta}{2}$, so the failure probability of Algorithm~\ref{alg: QuantumGetBounds} is $\delta$.

\textbf{Complexity.} The complexity of Step~4 is $O\rbra{\frac{\log\frac{1}{\delta}}{\beta}}$, and the complexity of Step~7 is $O\rbra{\frac{\log\frac{1}{\delta}}{\beta}}$. $O_\Omega$ can be constructed with  $O\rbra{1}$ uses of  $O_A$, $O_Y$ and $O_Z$ can be constructed with some controlled quantum gates. Therefore, the complexity of Step~10 is $O\rbra{\frac{\log\frac{1}{\delta}}{\beta}}$. The total number of queries to $O_A$ is $O\rbra{\frac{\log\frac{1}{\delta}}{\beta}}$.
\end{proof}

\textbf{Description of the main algorithm (Algorithm~\ref{alg: test-similar-means})}.    
 Combining Lemma~\ref{lem: property of far from} with Algorithm~\ref{alg: QuantumGetBounds}, we obtain our main procedure Algorithm~\ref{alg: test-similar-means}. By Lemma~\ref{lem: property of far from}, if the collection $\mathcal{D}$ is $\epsilon$-far from having $\gamma$-similar means, then there exists an interval $[x,y]$ of length $\gamma n$ such that there are distributions $D_1$ and $D_2$
satisfying
\begin{align}
    d(D_1,x) \ge \beta_1,\quad
    d(D_2,y) \ge \beta_2,
\end{align}
for suitable parameters $\beta_1,\beta_2=\Theta(\epsilon)$. Applying Algorithm~\ref{alg: QuantumGetBounds} to $D_1$ and $D_2$ produces estimates
$\tilde{\mu}(D_1^{>})$ and $\tilde{\mu}(D_2^{<})$ that approximate the corresponding extremal means under $\ell_1$ perturbations.  The resulting gap between $\tilde{\mu}(D_1^{>})$ and $\tilde{\mu}(D_2^{<})$ exceeds $\gamma n$, thereby serving as a witness that the collection does not satisfy the
$\gamma$-similar means property. The following theorem formalizes the performance guarantees of Algorithm~\ref{alg: test-similar-means}, establishing both its correctness and its query complexity under the quantum query access model.

\begin{theorem}\label{thm: upper bound of Oracle A}
For $\gamma,\epsilon\in\rbra{0,1}$ and a collection of $m$ distributions, there exists a quantum algorithm that can test $\gamma$-similar means in the quantum query access model, with probability at least $\frac{2}{3}$ using $\tilde{O}\rbra{\frac{1}{\epsilon}}$ queries.
\end{theorem}

\begin{proof}
    Since each call to Algorithm~\ref{alg: QuantumGetBounds} is made with failure probability at most $1/(3rt(q))$, the total failure probability is at most $1/3$. Consequently, throughout the remainder of the proof we assume that Algorithm~\ref{alg: QuantumGetBounds} is correct.

Assume that the collection $\mathcal{D}$ has $\gamma$-similar means. Theorem~
\ref{thm:l1qeme} implies $f_{q,h} \le \mu(D_{i^h_{q}}) \le e_{q,h}$. Hence, $\hat{y} - \hat{x}
\le
\max_i \mu(D_i) - \min_i \mu(D_i)
\le \gamma n$. Therefore, Algorithm~\ref{alg: test-similar-means} accepts.

Assume that the collection $\mathcal{D}$ is $\varepsilon$-far from having
$\gamma$-similar means.
By Lemma~\ref{lem: property of far from}, there exists an interval $[x,y]$ of length $\gamma n$
such that
\begin{align*}
    \sum_{i:\mu(D_i)<x} d_1(D_i,x)+\sum_{i:\mu(D_i)>y} d_1(D_i,y)=\frac{\epsilon m}{2}.
\end{align*}
We group the distributions satisfying $\mu(D_i)<x$ according to the value of $d_1(D_i,x)$, and define
\begin{align*}
    L_q:=\bigl\{ i : \mu(D_i)<x, d_1(D_i,x)\in\left(\beta_{q-1},\beta_q\right] \bigr\},
\end{align*}
where $\beta_q = (\varepsilon/8)2^q$. Analogously, for distributions with $\mu(D_i)>y$, we define the sets $R_q$.
It is easily proven that $\sum_{i:\mu\rbra{D_i}<x,d_1\rbra{D_i,x}\geq \frac{\epsilon}{8}}d_1\rbra{D_i,x}\geq \frac{\epsilon m}{8}$. By a counting argument, there exists an index $q_L$ such that $\sum_{i\in L_{q_L}}d_1\rbra{D_i,x}\geq \frac{\epsilon m}{8r}$, thus $|L_{q_L}|\geq \frac{\epsilon m}{8r\beta_q}=\frac{m}{r2^{q_L}}$.

In the same way, there exists an index $q_R$ such that $|R_{q_R}|\geq \frac{m}{r2^{q_R}}$. During iteration $q_L$, the algorithm selects $t(q_L)=\Theta(2^{q_L}\log(1/\varepsilon))$ indices from the collection, which suffices to guarantee that an index $i\in L_{q_L}$ is selected with high constant probability. An analogous argument shows that, in iteration $q_R$, an index $i'\in R_{q_R}$ is selected with high constant probability.

Theorem~\ref{thm:l1qeme} implies that for any $i\in L_{q_L}$, $e_{q_L,i} \le \mu^{>}_{\beta_{q_L}}(D_i) \le x$ and for any $i'\in R_{q_R}$, $f_{q_R,i'} \ge \mu^{<}_{\beta_{q_R}}(D_{i'}) \ge y$. So we have $\hat{x} \le x$ and $\hat{y} \ge y$, which implies $\hat{y}-\hat{x} \ge y-x = \gamma n.$ Hence, the algorithm rejects. 

The total complexity satisfies that
\begin{align}
    O\rbra{\sum_{q=1}^r t\rbra{q}\frac{\log\rbra{rt\rbra{q}}}{2^q\epsilon}}=\tilde{O}\rbra{\frac{1}{\epsilon}}.
\end{align}
\end{proof}

\subsection{Lower Bound}

As shown in Theorem~\ref{thm: upper bound of Oracle A}, similar means testing admits a quantum query complexity of $\tilde{O}\!\left(\frac{1}{\epsilon}\right)$. We next show that this dependence on $\epsilon$ is optimal up to logarithmic factors by reducing from the problem of distinguishing probability distributions. The corresponding lower bound, established in~\cite{belovs2019quantum}, is stated below.

\begin{algorithm}[H]
\caption{Quantum testing similar means QTSM$\rbra{O_A,\gamma,\epsilon}$}
\label{alg: test-similar-means}
\begin{algorithmic}[1]
    \STATE {\bfseries Input:} Quantum Oracle $O_A$ to a collection $\mathcal{D}=(D_1,\ldots,D_m)$, parameters $\gamma,\epsilon\in\rbra{0,1}$.
	\STATE {\bfseries Output:} Accept or reject.
    \STATE $r\gets \lceil \log(8/\varepsilon) \rceil$
	\FOR{$q = 1$ to $r$}
    \STATE Select $I_q = \{ i^1_{q}, \ldots, i^{t(q)}_{q} \} \subseteq [m]$
where $|I_q| = t(q) = \Theta(2^q \log(1/\varepsilon))$.
    \FOR{$h=1$ to $t_q$}
        \STATE $e_{q,h},f_{q,h}\gets l_1\text{QFME}\left(i^h_{q}, \frac{2^{q-1}\epsilon}{8}, \frac{1}{3rt(q)}\right)$.
    \ENDFOR
\ENDFOR
\STATE $\hat{x}\gets \min_{q,h} e_{q,h}$ and $\hat{y} = \max_{q,h} f_{q,h}$.
\IF{$\hat{y} - \hat{x} > \gamma n$}
\STATE \textbf{Reject}
\ELSE
\STATE \textbf{Accept}
\ENDIF 
\end{algorithmic}
\end{algorithm}

\begin{lemma}[Theorem 4 of \cite{belovs2019quantum}]\label{lem: distinguishing distributions}
   Let $p$ and $q$ be probability distributions over $[n]$.
Let $U_p$ and $U_q$ be purified quantum query access oracles for the distributions $p$ and $q$, respectively, satisfying
\begin{align*}
    U_p \ket{\vec{0}}
    = \sum_{i=1}^{n}\sqrt{p_i}\ket{i}\ket{\psi_i}, \\
    U_q \ket{\vec{0}}
    = \sum_{i=1}^{n}\sqrt{q_i}\ket{i}\ket{\phi_i}.
\end{align*}
Then any quantum algorithm that can distinguish between $p$ and $q$
requires $\Omega\!\left( \frac{1}{d_H(p,q)} \right)$ queries to $U_p$ and $U_q$, where $d_H(p,q)$ denotes the Hellinger distance, defined as $d_H(p,q)
= \sqrt{\frac{1}{2}\sum_{i\in[n]} (\sqrt{p_i}-\sqrt{q_i})^2 }$.
\end{lemma}

We now establish a tight query lower bound up to logarithmic factors for testing $\gamma$-similar means, obtained through a reduction from the problem of distinguishing two distributions.

\begin{theorem}\label{thm: lower bound of Oracle A}
    Any quantum algorithm for testing $\gamma$-similar means requires at least $\Omega\rbra{\frac{1}{\epsilon}}$ queries to the quantum query access.
\end{theorem}

\begin{proof}
    Let $D_1(0)=D_1(n)=\tfrac12$, 
$D_2(0)=\tfrac12+\epsilon$, $D_2(n)=\tfrac12-\epsilon$, 
and $D_3(0)=\tfrac12-\gamma$, $D_3(n)=\tfrac12+\gamma$.
By definition, the pair $(D_1,D_3)$ has $\gamma$-similar means,
while the pair $(D_2,D_3)$ is at least $2\epsilon$-far from having
$\gamma$-similar means.

Distinguishing $D_1$ from $D_2$ reduces to testing $\gamma$-similar means 
on the pair $(D, D_3)$: when $D = D_1$ the pair has $\gamma$-similar means, 
whereas when $D = D_2$ it is $2\epsilon$-far from having $\gamma$-similar means. Since $d_H(D_1, D_2) \le 2\epsilon$, 
Lemma~\ref{lem: distinguishing distributions} implies that distinguishing 
$D_1$ from $D_2$ requires $\Omega(1/\epsilon)$ queries. Thus the same lower bound carries over to testing $\gamma$-similar means.
\end{proof}

\section{Results for Quantum Sampling Access Model}
In this section, we establish upper and lower bounds on the query complexity under the quantum sample access model. The upper bound is $\tilde{O}\rbra{\frac{\sqrt{m}}{\epsilon}}$, obtained by simulating $O_A$ using $O_B$, which enables a reduction to the query model algorithm. Unlike the query model, where the complexity is independent of the number of distributions $m$, the sampling model exhibits an explicit dependence on $m$. For the lower bound, we show $\Omega\rbra{m^{1/3}+\frac{m^{1/4}}{\epsilon}}$ via a reduction from the problem of uniformity testing.

\subsection{Upper Bound}
We first show how to simulate $O_A$ using $O_B$.

\begin{lemma}\label{lem: construct A by B}
$O_A$ can be deterministically constructed from $O_B$ using \(O(\sqrt{m})\) applications of $O_B$.
\end{lemma}

\begin{proof}
The key idea is to use \emph{Quantum Deterministic Amplitude Amplification} (Lemma~\ref{lem: deterministic amplitude amplification}) to transform $O_B$ into $O_A$.
      Fix an $i$. Let $\text{C}_i\text{NOT}$ be a unitary such that $\text{C}_i\text{NOT}=\ket{i}\bra{i}_1\otimes X_4+\rbra{I-\ket{i}\bra{i}}_1\otimes I_4$. Let $U=\text{C}_i\text{NOT}\cdot O_B$. Then
    \begin{align*}
    \begin{split}
        U\ket{\vec{0}}_{1,2,3,4}&=\frac{1}{\sqrt{m}}\ket{i}_1\ket{\Psi_i}_{2,3}\ket{1}_4\\&+\sqrt{\frac{m-1}{m}}\ket{\Phi}_{1,2,3}\ket{0}_4,
    \end{split}
    \end{align*}
    where 
    \begin{align*}
    \begin{split}
        \ket{\Psi_i}&=\sum_{j=0}^n\sqrt{D_i\rbra{j}}\ket{j}_2\ket{\varphi_{i,j}}_3,\\
        \ket{\Phi}&=\sum_{g\in\sbra{m},g\neq i}\frac{1}{\sqrt{m-1}}\ket{g}_1\\
    &\otimes\rbra{\sum_{j=0}^n\sqrt{D_g\rbra{j}}\ket{j}_2\ket{\varphi_{g,j}}_3}.
    \end{split}
    \end{align*}
    Since $m$ is known, with Lemma~\ref{lem: deterministic amplitude amplification}, we can deterministically prepare $\ket{i}\ket{\Psi_i}\ket{1}$ with $O\rbra{\sqrt{m}}$ uses of $O_B$. Therefore, based on the relation $\ket{i}\ket{\Psi_i}\ket{1} = O_A\ket{i}\ket{\vec{0}}$, $O_A$ is deterministically implementable from $O_B$ using $O(\sqrt{m})$ calls.
\end{proof}

By replacing the calls to $O_A$ in Algorithm~\ref{alg: test-similar-means} with this simulation, we obtain the quantum algorithm in the quantum sampling access model stated in Theorem~\ref{thm: upper bound of Oracle B}.
\begin{theorem}\label{thm: upper bound of Oracle B}
     For $\gamma,\epsilon\in\rbra{0,1}$ and a collection of $m$ distributions, there exists a quantum algorithm that can test $\gamma$-similar means in the quantum sampling access model, with probability at least $\frac{2}{3}$ using $\tilde{O}\rbra{\frac{\sqrt{m}}{\epsilon}}$ queries.
\end{theorem}

\begin{proof}

Algorithm~\ref{alg: test-similar-means} is designed for $O_A$. Under quantum sample access model, each query to $O_A$ can be simulated by $O(\sqrt{m})$ queries to $O_B$ using Lemma~\ref{lem: construct A by B}. Since Algorithm~\ref{alg: test-similar-means} uses $\tilde{O}(1/\epsilon)$ queries to $O_A$, the resulting query complexity under $O_B$ is $\tilde{O}(\sqrt{m}/\epsilon)$.

Moreover, the success probability remains at least $2/3$, as the simulation is exact and introduces no additional error.
\end{proof}

\subsection{Lower Bound}

We next obtain a lower bound for the quantum sample access model by reducing from uniformity testing. The corresponding lower bound on the quantum query complexity of uniformity testing was proved in~\cite{chen2025list} and is stated below.

\begin{lemma}[Theorem 2.3 in~\cite{chen2025list}]\label{chen2025listlemma}
    Let $p$ be a probability distributions over $[n]$. Let $\textbf{U}_n$ be the uniform distribution  over $[n]$. Let $U_p$ be the purified quantum query access oracle for the distribution $p$. Then any quantum algorithm that can test whether $p=\textbf{U}_n$ or $\Abs{p-\textbf{U}_n}\geq \epsilon$ requires $\Omega\rbra{n^{1/3}+\frac{n^{1/4}}{\epsilon}}$ queries to $U_p$. 
\end{lemma}

We now establish the lower bound for testing $\gamma$-similar means, obtained through a reduction from the problem of uniformity testing.

\begin{theorem}\label{thm: lower bound of Oracle B}
    Any quantum algorithm for testing $\gamma$-similar means requires at least $\Omega\rbra{m^{1/3}+\frac{m^{1/4}}{\epsilon}}$ queries in the quantum sampling access model.
\end{theorem}

\begin{proof}
    Let $P_\epsilon$ be a probability distributions over $[2m]$ such that $P_\epsilon(j)=\frac{1+\epsilon}{2m}$ for $j\in\cbra{1,...,m}$, $P_\epsilon(j)=\frac{1-\epsilon}{2m}$ for $j\in\cbra{m+1,...,2m}$. Let $\textbf{U}_{2m}$ be the uniform distribution over $[2m]$. Then by Lemma~\ref{chen2025listlemma} the lower bound of the quantum query complexity for distinguishing between $P_\epsilon$ and $\textbf{U}_{2m}$ is $\Omega\rbra{m^{1/3}+\frac{m^{1/4}}{\epsilon}}$.

    Consider two collections of distributions $\mathcal{D}=\rbra{D_1,\ldots,D_m}$ and
    $\mathcal{D'}=\rbra{D'_1,\ldots,D'_m}$ where for every $i\in\sbra{m}$ $D_i(0)=D_i(n)=\frac{1}{2}$  and $D'_i(0)=\frac{1+\epsilon}{2}$, $D'_i(n)=\frac{1-\epsilon}{2}$. It is easy to verify that $\mathcal{D}$ has similar means whereas $\mathcal{D'}$ is $\epsilon$-far from this property. Given the quantum sampling access to $\mathcal{D}$ and $\mathcal{D'}$, we can simulate the purified quantum query access oracles of $P_\epsilon$ and $\textbf{U}_{2m}$. Therefore, the lower bound of distinguishing between $P_\epsilon$ and $\textbf{U}_{2m}$ implies the lower bound of testing similar means,  namely $\Omega\rbra{m^{1/3}+\frac{m^{1/4}}{\epsilon}}$.
\end{proof}

\section{Conclusions}

In this work, we developed quantum algorithms for testing mean similarity of multiple distributions. Our algorithms achieve provable quantum speedups in both the query and sampling access models, accompanied by nearly matching lower bounds.

Several questions remain open. In particular, the dependence on the number of distributions $m$ in the sampling model is not yet fully understood. It would also be interesting to extend quantum techniques to other multi-distribution testing problems, such as uniformity and closeness testing.

\bibliographystyle{alpha}
\bibliography{tsm}
\end{document}